\documentclass[11pt]{article} % use larger type; default would be 10pt
\usepackage[round,authoryear]{natbib}
\usepackage[nottoc,notlof,notlot]{tocbibind} % Put the bibliography in the ToC

\usepackage{geometry} % to change the page dimensions
\usepackage{graphicx} % support the \includegraphics command and options
\usepackage{subfigure}
\usepackage{caption}
\usepackage{booktabs} % for much better looking tables
\usepackage{array} % for better arrays (eg matrices) in maths
\usepackage{paralist} % very flexible & customisable lists (eg. enumerate/itemize, etc.)
\usepackage{verbatim} % adds environment for commenting out blocks of text & for better verbatim
\usepackage{subfig} % make it possible to include more than one captioned figure/table in a single float
\usepackage{amssymb}
\usepackage{amsmath}
\usepackage{float}
\usepackage{graphicx} % Allows including images
\usepackage{booktabs} %
\usepackage{fancyhdr} % This should be set AFTER setting up the page geometry
\usepackage{sectsty}
\allsectionsfont{\sffamily\mdseries\upshape} % (See the fntguide.pdf for font help)
\usepackage[nottoc,notlof,notlot]{tocbibind} % Put the bibliography in the ToC
\usepackage[titles,subfigure]{tocloft} % Alter the style of the Table of Contents
\numberwithin{equation}{section}
\newtheorem{rem}{Remark}[section]
\newtheorem{thm}{Proposition}[section]

\newtheorem{proof}{Proof}[section]

\title{Penalized EM algorithm and copula skeptic graphical models for inferring networks for mixed variables}

\author{Fentaw Abegaz and Ernst Wit\\
 Johann Bernoulli Institute of Mathematics and Computer Science \\
 University of Groningen\\%
 }
\date{}

\begin{document}
\maketitle
\begin{abstract}

In this article, we consider the problem of reconstructing networks for continuous, binary, count and discrete ordinal variables by estimating sparse precision matrix in Gaussian copula graphical models.  We propose two approaches: $\ell_1$ penalized extended rank likelihood with Monte Carlo Expectation-Maximization algorithm (copula EM glasso) and copula skeptic with pair-wise copula estimation for copula Gaussian graphical models. The proposed approaches help to infer networks arising from nonnormal and mixed variables. We demonstrate the performance of our methods through simulation studies and analysis of breast cancer genomic and clinical data and maize genetics data. 

\end{abstract}

\section{Introduction}
\label{sec:introduction}

The aim of this article is to formulate an inference approach for the analysis of high dimensional data  that involves mixed variables of continuous, binary and ordered categorical types using graphical models. In particular, we focus on model estimation and identification of undirected graph structure for Gaussian graphical models for high dimensional datasets. We base our inference procedure on the EM algorithm and pair-wise copula estimation with $\ell_1$ penalized extended rank likelihood. 

Toward the study of mixed variables determination of their joint distribution is the main challenge. The seminal work of \cite{sklar1959fonctions} that formally introduced the notion of copula provide the theoretical framework, in which a joint probability distribution can be represented by its univariate marginal distributions and a copula function. As a result multivariate association, which is fully described by the copula function, can be modeled separately from the univariate marginal distributions.

In copula modeling, \citet*{genest1995semiparametric} developed a popular semiparametric estimation or ``rank likelihood'' based estimation, in which the association among the variables are represented by a parametric copula model but the marginals are treated as nuisance parameters and estimated nonparametrically. The resulting semiparametric estimators are well-behaved for continuous data but fail for discrete data, for which the distribution of the ranks depends on the univariate marginal distributions, making them somewhat inappropriate for the analysis of mixed continuous and discrete data \citep{hoff2007extending}. To remedy this, \citet{hoff2007extending} propose the extended rank likelihood, which is a type of marginal likelihood that does not depend on the marginal distributions of the observed variables. Under the extended rank likelihood approach, the ranks are free of the nuisance parameters (or marginal distributions) of the discrete data. This makes the extended rank likelihood approach more suited for the determination of graphical models in the mixed variable setting and avoids the difficult problem of modeling marginal distributions \citep{dobra2011copula}.

The extended rank likelihood estimation is implemented for the study of association among mixed variables under a Bayesian framework by \citet{hoff2007extending} and further studied in the graphical model setting by \citet{dobra2011copula} using Bayesian model averaging approach for graph identification and estimation in copula Gaussian graphical models. Since the marginals are treated as nuisance parameters, the parameter of interest for estimation is the correlation matrix or the precision matrix, i.e. the inverse of the correlation matrix in case of a Gaussian copula. \citet*{ambroise2009inferring} raised their concern on the challenging task involved in the Bayesian framework to construct priors on the set of precision or concentration matrices. In this article we propose an alternative approach that consider the extended rank likelihood under $l_1$ penalized maximum likelihood setting with the Expectation-Maximization (EM) algorithm for high-dimensional inference based on graphical models. This approach is referred to as copula EM glasso. 

On the other hand, \citet{liu2012high} considered graphical modeling for binary and continuous variables using nonparanormal distributions and glasso algorithm of \citet*{friedman2008sparse}. In particular, in their nonparanormal skeptic approach, they  considered the rescaled empirical distribution transformation of data (with or without truncation and monotone transformation) to compute correlation matrix based on nonparametrically estimated pairwise rank correlations. We observe that the use of one step glasso algorithm makes their approach computationally efficient for high dimensional setting. Further we note that rank correlations such as Kendall's tau and Spearman's rho are directly related to bivariate copula models. Through these relationships and upon carefully selected bivariate copulas more accurate estimation of the rank correlations can be achieved. Thus, we extend the paranonnormal skeptic approach for rank correlations computed from bivariate parametric copulas. This approach is referred to as copula skeptic glasso.

We apply the proposed approaches to breast cancer genomic and clinical data. Breast cancer is the leading cause of death among women in the world and represents a significant health problem. Multiple factors like age, diet, obesity, parity, age at first childbirth, oral contraceptives, exogenous estrogens, genetics, environment, geographic location influence the development of breast cancer. However, the majority of the cases in breast cancer is always due to genetic abnormalities. At present, only small numbers of accurate prognostic and predictive factors are used clinically for managing the patients with breast cancer \citep{kumar2012application}. In the last few decades knowledge of breast cancer grade determined by Nottingham prognostic index (NPI) has been very helpful to decide on the most effective treatments. Moreover, microarray-based gene expression profiling has been used extensively to characterize the transcriptome of breast cancer, resulting in the identification of new molecular subtypes and markers or signatures of potential therapeutic and prognostic importance \citep{ringner2011gobo}. Inclusion of such treatment predictive markers considerably improved breast cancer treatment decisions.  To further tailor treatment for individual patients, identification of additional clinical and genetic markers is required. 

Genomic DNA copy number alterations, i,e., amplifications or deletions, are key genetic events in the development and progression of breast cancers. Gene copy number changes can be determined on a gene-by-gene basis using microarrays. A genome-wide microarray comparative genomic hybridization (CGH) is used to analyse the pattern of DNA copy number alteration with the aim to study the relationship between DNA amplification and deletion patterns and severity of breast cancer as measured by several clinical indicators on patients. Details of the experiment is discussed in \citet{jensen2009frequent}. The data from the breast cancer experiment include 296 variables of which 287 are genes and 9 are clinical variables obtained from 106 breast cancer patients. The genomic and clinical variables are mixed measurements of continuous, binary and ordered categorical types, see details in Section \ref{sec:analysisofdata}.

We also considered a second application of our approach on a very high dimensional setting on data from maize genetic nested association mapping population that has been analysed and discussed by \citet{mcmullen2009genetic}. The data set includes 1106 SNP loci or genetic markers and 4699 replicates.  Our objective is to obtain a sparse representation of potential trans-acting genetic markers that may provide information for a better understanding of the molecular basis of phenotypic variation. This helps to improve agricultural efficiency and sustainability.

The rest of this article is organized as follows. Section \ref{sec:graphicalmodels} provides a brief description of Gaussian copula modeling aspects related to continuous, binary, count and ordinal variables. Section \ref{sec:l1penalizedemestimation} formulates the copula based EM algorithm with $l_1$ penalized likelihood estimation and the copula skeptic glasso with pair-wise copula selection. It also discusses the selection of tuning parameter in case of EM formulation of glasso. Section \ref{sec:analysisofdata} demonstrates performance of the proposed approach using simulation studies and the analysis of high dimensional data on breast cancer and maize genetic properties. We close with a concluding remark in Section \ref{sec:conclusion}. 

\setcounter{equation}{0} \setcounter{figure}{0}
\section{Copula Graphical models}
\subsection{Gaussian copula graphical models}
\label{sec:graphicalmodels}

Graphical models are efficient tools for studying multivariate distributions through a compact, graphical representation of the joint probability distribution of the underlying random variables. The treatment of graphical models simplifies significantly, when one focuses on normally distributed variables. Let the random vector $Y=(Y_{1}, \ldots, Y_{p})^T$  be assumed to be Gaussian with a positive-definite covariance matrix $\Sigma$ of dimension $p \times p$. A graphical model $G = (V, E)$, where $V$ corresponds to the set of nodes with $p$ elements and $E \subset V\times V$ of ordered pairs of distinct nodes called the edges of $G$, for $\mathcal{N}_{p}(0, \Sigma)$ is called a Gaussian graphical model, if on the graph $G$, the edges $E$ represent conditional dependence among the random variables. Absence of an edge between any pair of random variables or zero value of a precision matrix, $\Theta = \Sigma^{-1}$, corresponds to conditional independence of the two random variables given the remaining ones.

In practice, we encounter both discrete and continuous variables that may not be Gaussian. Thus, the assumption of multivariate normal distribution would be too restrictive. To relax the normality requirement, we use the copula framework to construct multivariate distributions for arbitrary marginals as discussed in Section \ref{sec:introduction} above. In order to use the properties of the Gaussian graphical model, we consider the Gaussian copula. The Gaussian copula with correlation matrix $\Gamma$ of dimension $p \times p$ having $p(p-1)/2$ free parameters is given by:
 \begin{equation*}
    C(u_1, \ldots, u_{p} \mid \Gamma) = \Phi_{p}(\Phi^{-1}(u_1), \ldots, \Phi^{-1}(u_{p}) \mid \Gamma),
   \end{equation*}
and the corresponding Gaussisn copula-based distribution function is
 \begin{eqnarray}
    H(y_1, \ldots, y_p | \Gamma, F_1, \ldots, F_{p}) = \Phi_{p}(\Phi^{-1}(F_1(y_1)), \ldots, \Phi^{-1}(F_{p}(y_{p}))  \mid \Gamma).  \label{gc}
    \label{eq:gaussiancopulaGM}
\end{eqnarray}
Here $\Phi(\cdot)$ represents the CDF of the standard normal distribution and $\Phi_{p}(\cdot \mid \Gamma)$ is the CDF of
multivariate normal distribution $\mathcal{N}_P(0, \Gamma)$. 

We note that under the Gaussian copula the correlation matrix $\Gamma$ is the matrix of correlation coefficients among the transformed variables
$Z_j = \Phi^{-1}(F_{j}(y_{j}))$,  $j=1, \ldots, p$, which represent the maximum pairwise correlations among the $Y_{j}$s, $j=1, \ldots, p$. If the univariate marginal distributions are normal, then entries of the correlation matrix represent exactly the pairwise correlation coefficients of the variables.

\begin{thm}
 Let $\Gamma$ be a positive-definite matrix, such that $Z \sim \mathcal{N_P}(0,\Gamma)$ is a graphical model with respect to graph $G=(V,E)$. Then the continuous variable $Y$, defined via \ref{eq:gaussiancopulaGM}, is also a graphical model with respect to $G$. In particular, the precision matrix $\Theta = \Gamma^{-1}$  represents the conditional independence among the observed variables $Y_{j}$s.
\end{thm}

\begin{proof} This is as a result of the invariance property of conditional independence relation over equivalent probability measures as shown in \citet[Theorem 3.6]{van1985invariance}.
\end{proof}

We now focus on graphical modeling for observed variables $Y$ of mixed type, i.e. in the case that they represent a collection of continuous, binary, ordinal or count variables. Suppose the $j$-th variable $Y_{j}$ has marginal distribution $F_{j}$ with its pseudo-inverse $F_{j}^{-1}$. In copula modeling, the marginals are treated as nuisance parameters and estimated nonparametrically mainly using the rescaled empirical distribution:  $\hat{F}_{j}(y) = \frac{1}{n+1}\sum_{i=1}^n \mathcal{I}\{Y_{ji} \leq y\}$, $j=1, \ldots, p$. A copula graphical model, discussed above, can be constructed by introducing a vector of latent variables $Z \sim \mathcal{N_p}(0,\Gamma)$ that are related to the observed variables $Y$ as $Y_{j} = \hat{F}_{j}^{-1}(\Phi(Z_j))$, $j=1, \ldots, p$. In the case of mixed variables, the graphical structure, i.e. the conditional independence implied by the graph structure, is assumed to hold exclusively on the latent variable $Z$.

The aim of the inference procedure is to infer the graphical structure $G$, defined by the latent variable $Z$. Though the $Z$s are not observable, the observed $Y_{j}$s do provide a limited amount of information about them. Since the $\hat{F}_{j}$s are nondecreasing, observing $Y_{k_1} < Y_{k_2}$ implies that $Z_{k_1} < Z_{k_2}$. More generally, observing the n-dimensional vector $Y_i$ tells us that $Z_i$ lies in
\begin{equation}
 D(Y_i) = \{z \in \Re^{n} | a_{i}(Y_{ij}) < Z_j \leq b_{i}(Y_{ij})) \}, \label{hs}
\end{equation}
where $a_{i}(y) = \inf \{z| \hat{F}_j^{-1}(\Phi(z)) = y\}$ and $b_{i}(y) = \sup \{z| \hat{F}_j^{-1}(\Phi(z)) =  y \}$. In fact, for every ordinal $Y_j$, we can identify a set of thresholds $\tau_j=\left(\tau_{j0}, \tau_{j1}, \ldots, \tau_{j{n_j}}\right)$ with
\[ -\infty=\tau_{j0} < \tau_{j1} < \cdots < \tau_{j{n_j}} = \infty, \]
such that for some ordered set of values $\left\{ c_{j1} < \cdots < c_{j{w_j}} \right\}$,
	     \begin{eqnarray*}
	         Y_{j} = \sum^{n_j}_{r=1} c_{jr} \times I\left\{{\tau_{j,{r-1}} < z_{j} \leq \tau_{jr}}\right\}.
	     \end{eqnarray*}
It follows that the mapping of the ordered values of $Y_j$ into some defined intervals $(\tau_{jr},\tau_{j,r+1}]$ of the latent variable $Z_j$ relies on the following
relationship
\[  \tau_{jr} = \Phi^{-1}(\hat{F}_{j}(c_{jr}))),~~ r=1,\ldots,n_j-1.\]
The collection of these intervals is the set $D(Y)=D(Y_1, \ldots, Y_n)$ in \eqref{hs}. In case of missing observations, we consider data missing in this study as missing completely at random so that the missing values are easily determined from the latent variable distribution defined on the interval $(-\infty, \infty)$. Then the occurrence of event $Z \in D(Y)$ is taken as the observed data to infer about the copula parameter or the graph structure separately from the marginal distributions. Such inference approach is similar to the extended rank likelihood in \citet{hoff2007extending} and copula Gaussian graphical modeling in \citet{dobra2011copula}.

\setcounter{equation}{0} \setcounter{figure}{0}

\section{Sparse inference methods}
\subsection{ Copula EM GLASSO approach}
\label{sec:l1penalizedemestimation}
The Expectation-Maximization (EM) algorithm \citep*{dempster1977maximum} is a popular method for maximum likelihood estimation in the case of incomplete data, which naturally occur in our setting as a result of the latent nature of $Z$ as discussed in the previous section. In this section we consider EM algorithm with penalized likelihood. \citet{green1990use} studied convergence properties of the EM algorithm for penalized likelihood. 

The marginal likelihood of $Y$ where we consider $F_{1}, \ldots, F_{p}$ as nuisance parameters; see also \citet{hoff2007extending}, is
\begin{eqnarray}
 L_Y(\Theta) =  \int_{D(Y)} p(z  \mid \Theta ) d z \label{likH}
\end{eqnarray}

For large sample sizes the precision matrix $\Theta$ can be estimated by maximizing the log-likelihood $l(\Theta)$ as a function of $\Theta$. Whereas for high-dimensional data $(n << p)$, we add an $l_1$-norm penalty to encourage sparsity in the precision matrix and the identification of the underlying graph. The $l_1$ penalized log-likelihood is given by
\begin{eqnarray}
   \ell_\lambda(\Theta) =  \log L_Y(\Theta) - \lambda \left\|\Theta\right\|_1  , \label{penL}
\end{eqnarray}
where the scalar parameter $\lambda \geq 0$ controls the size of the penalty.

Due to the complexity of maximizing the marginal log-likelihood $l_Y(\Theta)$ in \eqref{likH} and \eqref{penL}, we employ the EM algorithm. We have discussed in the previous section that the observed data $Y$ provide some information on the latent variables $Z$ such that the occurence of the event $Z \in D(Y)$ is taken as the observed data to infer on $\Theta$. We now recall from standard EM algorithm setting that the loglikelihood of the observed data can be expressed as 
\begin{equation}
   \ell(\Theta) = Q(\Theta \mid \Theta^{(m)}) - H(\Theta \mid \Theta^{(m)}), \label{obsloglik}
\end{equation}   
where $\Theta^{(m)}$ an estimate of $\Theta$ from the previous step of the algorithm,
\begin{equation}
   Q(\Theta \mid \Theta^{(m)}) =E\left[\log L_{Z, Z \in D(Y)}(\Theta) \mid Z \in D(Y), \Theta^{(m)} \right] \label{hoffeqn},
\end{equation}
and
\begin{equation}
   H(\Theta \mid \Theta^{(m)}) = E\left[\log L_{Z \mid Z \in D(Y)}(\Theta) \mid Z \in D(Y), \Theta^{(m)} \right].
\end{equation}
The penalized log-likelihood takes the form
\begin{eqnarray}
   \ell_\lambda(\Theta) =  Q(\Theta \mid \Theta^{(m)}) - H(\Theta \mid \Theta^{(m)}) - \lambda \left\|\Theta\right\|_1 , \label{penLQH}
\end{eqnarray}
such that
\begin{eqnarray}
   \Theta_\lambda^{({m+1})} = \textrm{argmax}_{\Theta} \left\{ Q(\Theta \mid \Theta^{(m)}) - H(\Theta \mid \Theta^{(m)}) - \lambda \left\|\Theta\right\|_1 \right\}. \label{maxp1}
\end{eqnarray}
Further, from standard EM approach,  $H(\Theta \mid \Theta^{(m)}) \leq H(\Theta^{(m)} \mid \Theta^{(m)}) $
for any $\Theta$ in the parameter space. Thus, obtaining an updated estimate of the parameter by maximizing the $l_1$ penalized log-likelihood in \eqref{maxp1} reduces to
\begin{eqnarray}
   \Theta_\lambda^{({m+1})} = \textrm{argmax}_{\Theta} \left\{ Q(\Theta \mid \Theta^{(m)}) - \lambda \left\|\Theta\right\|_1 \right\}. \label{maxp2}
\end{eqnarray}

The EM optimization strategy alternates iteratively between the E-step, computing conditional expectation of the complete log-likelihood $Q(\Theta \mid \Theta^{(m)})$ and the M-step, maximizing $Q(\Theta \mid \Theta^{(m)})$, with a sparsity penalty $\lambda \left\|\Theta\right\|_1$, over $\Theta$. 

{\it{E-step}}: The complete data likelihood depends on the joint distribution of $(Z,Z \in D(Y))$ given by
\begin{eqnarray*}
p(Z,Z \in D(Y) \mid \Theta) =  \left\{
\begin{array}[pos] {ll}%{spalten}
	p(Z \mid \Theta) ~~~~~~~~~~~~~~~~~~~~ Z \in D(Y) \\
	0 ~~~~~~~~~~~~~~~~~~~~~~~~~~~~~ Z \notin D(Y) 
\end{array} 
\right.
\end{eqnarray*}
where $p(Z \mid \Theta)$ is the multivariate normal density with mean zero and variance $\Sigma = \Theta^{-1}$.
Then the complete log likelihood of $(Z,Z \in D(y))$ for a random sample of size $n$ after ignoring constants with respect to $\Theta$ is given by
\begin{eqnarray}
l_{Z}(\Theta) &=&  \sum_{i=1}^n \log \left(p(Z_i \mid \Theta)\right) I_{\{Z_i \in D(Y_i)\}}  \nonumber\\
&=& -\frac{np}{2} \log(2 \pi) + \frac{n}{2} \log \det(\Theta) - \frac{1}{2} \sum^n_{i=1} Z_i^T \Theta Z_i I_{\{Z_i \in D(Y_i)\}} \label{complik}
\end{eqnarray}
Using the complete log likelihood in \eqref{complik}, it follows that
 \begin{eqnarray}
    Q(\Theta \mid \Theta^{(m)})  &=& E\left[l_{Z}(\Theta) \mid Z \in D(Y), \Theta^{(m)} \right] \nonumber \\
     &=&  \frac{n}{2}\left\{ \log \det(\Theta) - \frac{1}{n} \sum_{i=1}^n \left(  E\left[ Z_i^T \Theta Z_i \mid  Z_i \in D(Y_i), \Theta^{(m)}\right]\right) \right\} \nonumber \\
     &=&  \frac{n}{2}\left\{ \log \det(\Theta) -  \mbox{Tr}\left( \Theta \frac{1}{n} \sum_{i=1}^n E\left[ Z_i Z_i^T \mid  Z_i \in D(Y_i), \Theta^{(m)}\right]\right) \right\} \nonumber \\
   &=& \frac{n}{2} \left\{ \log \det(\Theta) - \mbox{Tr}\left(\Theta \bar{R}\right) \right\},  \label{QEM1}
  \end{eqnarray}
where $\mbox{Tr}$ stands for the trace of a matrix and $\bar{R}$ is the expected empirical covariance function of the latent variables given $Z \in D(Y)$:
 \[ \bar{R} = \frac{1}{n} \sum_{i=1}^n E\left[Z_i Z_i^T \mid z_i \in D(Y_i), \Theta^{(m)}\right], \]
where
 \begin{eqnarray}
 E\left[Z_i Z_i^T \mid Z_i \in D(Y_i), \Theta^{(m)}\right] &=& E\left[Z_i  \mid Z_i \in D(y_i), \Theta^{(m)}\right] E\left[Z_i  \mid Z_i \in D(y_i), \Theta^{(m)}\right]^T  \nonumber \\
     & &  + ~~ cov \left[Z_i \mid Z_i \in D(y_i), \Theta^{(m)} \right]  \label{condmv}
 \end{eqnarray}
Note that the conditional latent random variable $\left\{Z \mid Z \in D(Y)\right\}$ follows a truncated multivariate normal distribution. Analytical expressions for the computations of moments for truncated multivariate normal distribution are given in  \citet{wilhelm2010tmvtnorm} and references therein. However, due to the computational complexity, obtaining analytical solutions is only feasible for very few variables. Another approach is to simulate a large sample from the truncated multivariate normal distribution and calculate the sample conditional covariance matrix and sample conditional mean to estimate $E\left[Z_i Z_i^T \mid Z_i \in D(y_i), \Theta^{(m)}\right]$ using \eqref{condmv}. 

Alternatively, towards a computational efficient approach instead of mapping all mixed variables to a latent space as discussed above we may partition the mixed variables into two as continuous denoted by $Y_c$, and ordered variables that includes ordinal, binary and counts denoted by $Y_o$. We also partition the correlation matrix along with the variables grouping as
\begin{eqnarray*}
\Sigma = 
\begin{bmatrix}
\Sigma_{cc} & \Sigma_{co} \\
\Sigma_{oc} & \Sigma_{oo}
  \end{bmatrix}.
\end{eqnarray*}

We then take $Z = (Z_c, Z_o) \sim N\left(0, \Theta \right)$, where $\Theta = \Sigma^{-1}$ with $Z_c = \Phi^{-1}(\hat{F}(Y_c))$ is   transformed normal scores of observed continuous variables using the rescaled empirical distribution based on the natural Gaussian copula semiparametric approach and $Z_o \in D(Y_o)$ is the latent normal score corresponding to the ordered observed variables $Y_o$ obtained in a similar way as discussed in Section \ref{sec:graphicalmodels}
.

With a similar argument as above the complete data likelihood depends on the joint distribution: $p(Z_c,Z_o, Z_o\in D(Y_{o}) \mid \Theta) = p(Z_c,Z_o)$, for $Z_o \in D(Y_o)$. Such that the complete data loglikelihood after ignoring constants is given by
\begin{eqnarray}
l_{Z_c,Z_o}(\Theta) &=&  \sum_{i=1}^n \log \left(p(Z_{c_i}, Z_{0_i} \mid \Theta)\right) I_{\{Z_{o} \in D(Y_{o})\}}  \nonumber\\
&=&  \frac{n}{2} \log \det(\Theta) - \frac{1}{2} \sum^n_{i=1} \left[Z_{c_i},Z_{o_i} \right] ^T \Theta \left[Z_{c_i},Z_{o_i} \right] I_{\{Z_{o_i} \in D(Y_{o_i})\}} \label{complik2}
\end{eqnarray}
Using this complete data log-likelihood and after ignoring constants with respect to $\Theta$, it follows that
 \begin{eqnarray}
    Q(\Theta \mid \Theta^{(m)})  &=& E_{Z_o}\left[l_{Z_c, Z_o }(\Theta) \mid Y_c, Z_o \in D(Y_o), \Theta^{(m)} \right] \nonumber \\
     &=&  \frac{n}{2}\left\{ \log \det(\Theta) \right.\\
		& & ~~ \left. -  \mbox{Tr}\left( \Theta \frac{1}{n} \sum_{i=1}^n E_{Z_o}\left[ \left[Z_{c_i},Z_{o_i} \right] \left[Z_{c_i},Z_{o_i} \right]^T \mid  Y_{c_i}, Z_{o_i} \in D(Y_{o_i}), \Theta^{(m)}\right]\right) \right\} \nonumber \\
   &=& \frac{n}{2} \left\{ \log \det(\Theta) - \mbox{Tr}\left(\Theta \tilde{R}\right) \right\},  \label{QEM2}
  \end{eqnarray}
where 
\begin{eqnarray}
  \tilde{R} = \frac{1}{n} \sum_{i=1}^n E_{Z_o}\left[ \left[Z_{c_i},Z_{o_i} \right] \left[Z_{c_i},Z_{o_i} \right]^T \mid  Y_{c_i}, Z_{o_i} \in D(Y_{o_i}), \Theta^{(m)}\right].  
\end{eqnarray}
 An estimate of $\tilde{R}$ can be obtained after evaluating the expectations as follows. 
\begin{eqnarray*}
& E_{Z_o}\left[ Z_{c_i} Z_{c_i}^T \mid  Y_{c_i}, Z_{o_i} \in D(Y_{o_i}), \Theta^{(m)}\right] = Z_{c_i} Z_{c_i}^T \\
& E_{Z_o}\left[ Z_{c_i} Z_{o_i}^T \mid Y_{c_i}, Z_{o_i} \in D(Y_{o_i}), \Theta^{(m)}\right] = Z_{c_i} \hat{Z}_{o_i}^T \\
& E_{Z_o}\left[ Z_{o_i} Z_{o_i} ^T \mid  Y_{c_i}, Z_{o_i} \in D(Y_{o_i}), \Theta^{(m)}\right] = \hat{Z}_{o_i} \hat{Z}_{o_i}^T + \Theta^{{(m)}^{-1}}_{oo}. 
\end{eqnarray*}
where  $\hat{Z}_{o_i}^T= E_{Z_o}\left[ Z_{o_i} \mid  Z_{o_i} \in D(Y_{o_i}), \Theta^{(m)}\right] - \Theta^{{(m)}^{-1}}_{oo} \Theta^{(m)}_{oc} Z_{c_i}$,  is a conditional expectation defined on the distribution $p(Z_{o} \mid  Y_{c})$ for $Z_{o} \in D(Y_{o})$.

{\it{M-step}}: This involves updating the parameter $\Theta$ using the $l_1$ penalized log-likelihood given by
\begin{eqnarray}
   \Theta_\lambda^{({m+1})} = \textrm{argmax}_{\Theta} \left\{ Q(\Theta \mid \Theta^{(m)}) - \lambda \left\|\Theta\right\|_1 \right\}. 
\end{eqnarray}
We next substitute the $Q$ function by \eqref{QEM1} or \eqref{QEM2} from the E-step to obtain
\begin{eqnarray}
   \Theta_\lambda^{({m+1})} = \textrm{argmax}_{\Theta} \left\{ \log \det(\Theta) - \mbox{Tr}( \Theta \bar{R}) - \lambda \left\|\Theta\right\|_1 \right\}. \label{maxl1pen}
\end{eqnarray}
The maximization problem in \eqref{maxl1pen} takes the form of $l_1$ penalized likelihood for Gaussian graphical models and computation is done efficiently using the graphic lasso algorithm \citep{friedman2008sparse}. This algorithm is fast and allows the re-use of the estimate under one value of the tuning parameter as a ``warm"' start for the next value. The determination of a value for $\lambda$ in case of penalized inference with EM algorithm is discussed in Section \ref{sec:modelselection}.

\begin{rem}
Setting the penalty parameter $\lambda = 0$ results in the unpenalized maximum likelihood estimate which can be considered as an alternative to the Bayesian approach discussed in \citep{hoff2007extending}.
\end{rem}

\subsection{Copula skeptic GLASSO}

The copula EM glasso approach discussed in the previous section, though it is a natural approach, it is computationally expensive, since it calls MCMC in the E-step and glasso in the M-step. In particular, in a very high dimensional setting the computational issue requires further attention. One approach is to seek a one-to-one mapping of the observed and latent variables that avoids the Monte Carlo EM algorithm resulting from a one-to-many mapping. The semiparametric copula approach of estimating marginals through the rescaled empirical distribution is a one-to-one mapping or transformation of the observed data. Instead of directly using the transformed data to estimate $\Theta$, a sample correlation matrix can be compute from pairwise rank correlations. In this regard, \citet{liu2012high}  considered a nonparanormal skeptic approach to obtain sparse estimates of $\Theta$ for binary and continuous variables using one step glasso based on the estimated correlation matrix. 

We note that the use a one step glasso approach is computationally efficient. Further, we note that rank correlations like Kendall's tau and Spearman's rho can be better approximated by a carefully chosen parametric bivariate copula model that takes in to account the underlying bivariate dependence structure. The vast literature on copulas deals with bivariate copula models and has demonstrated their potential to capture various types of dependence structure. 

It is known, for example, that the population version of Kendall's tau is related to parametric copula models parametrized by  $\gamma_{ij}$ via 
\begin{eqnarray*}
   \tau_{ij} = 4 \int^1_0 \int^1_0 C(u,v \mid \theta_{ij})dC(u,v \mid \gamma_{ij}) - 1 . 
\end{eqnarray*}
For commonly used copula models, there is closed form representation of the Kendall's tau using the bivariate copula parameter, see for example \citet{nelsen2006introduction}. An estimate of Kendall's tau is obtained using   
\begin{equation*}
\hat{\tau}_{ij}=%
\begin{cases}
 4 \int^1_0 \int^1_0 C(u,v \mid \hat{\theta}_{ij})dC(u,v \mid \hat{\gamma}_{ij}) - 1  & \text{for}\quad i \neq j \\ 
1 & \text{for}\quad i = j %
\end{cases}%
\end{equation*}
or its closed form representation. Further Kendall's tau is related to the correlation coefficient, $\Gamma$,  by 
\begin{equation*}
\hat{\Gamma}_{ij}=%
\begin{cases}
\sin \left(\frac{\pi}{2} \hat{\tau}_{ij}\right) & \text{for}\quad i \neq j \\ 
1 & \text{for}\quad i = j %
\end{cases}%
\end{equation*}

Then to obtain sparse estimates, glasso can be implemented that uses the estimated correlation matrix $\hat{\Gamma}$ in the direct optimization of the objective function:
\begin{eqnarray*}
   \hat{\Theta}_\lambda = \textrm{argmax}_{\Theta} \left\{ \log \det(\Theta) - \mbox{Tr}( \Theta \hat{\Gamma}) - \lambda \left\|\Theta\right\|_1 \right\}. 
\end{eqnarray*}  

\subsection{Model selection}
\label{sec:modelselection}
\setcounter{equation}{0}
 \setcounter{figure}{0}
 \setcounter{table}{0}
For high dimensional data, the empirical covariance matrix is singular and poses computational problems. However, our $l_1$ penalized approach guarantees with probability one a positive definite precision matrix with the additional property of being sparse. Note that sparseness refers to the property that all parameters that are zero are actually estimated as zero with probability tending to one. This helps to assess conditional independence based on entries of the precision matrix \citep{dempster1972covariance}. 

Under the $l_1$ penalized maximum likelihood setting the sparsity of the estimated precision matrix is controlled by the penalty parameter $\lambda$ in  \eqref{maxl1pen}. We follow information based criteria in order to obtain reasonably sparse precision matrix. One could also use cross-validation for the choice of $\lambda$ which we have not consider in this article.

We now consider \eqref{obsloglik} that suggests the log likelihood of the observed data can be computed at EM convergence, see for example \citet{ibrahim2008model}. Let the estimate $\widehat{\Theta}_{\lambda}$ is obtained at EM convergence for a given value of $\lambda$. The log likelihood of the observed data is
  \[ \log L_{Y}(\widehat{\Theta}_{\lambda}) = Q(\widehat{\Theta}_{\lambda} \mid \widehat{\Theta}_{\lambda}) -  H(\widehat{\Theta}_{\lambda} \mid \widehat{\Theta}_{\lambda}).\]
Thus a model selection criterion is defined by
 \begin{eqnarray*}
   IC(\lambda) &=& -2 \log L_{Y}(\widehat{\Theta}_{\lambda}) + pen(\widehat{\Theta}_{\lambda})  \\
      &=& -2 Q(\widehat{\Theta}_{\lambda} \mid \widehat{\Theta}_{\lambda}) +  2 H(\widehat{\Theta}_{\lambda} \mid \widehat{\Theta}_{\lambda}) + pen(\widehat{\Theta}_{\lambda}) ,
   \end{eqnarray*}
where $pen(\widehat{\Theta}_{\lambda})$ refers to a penalty term. Different forms of $pen(\widehat{\Theta}_{\lambda})$ lead to different model selection criteria. Let $ d  $ denotes the number of non-zero upper or lower off-diagonal elements of $\widehat{\Theta}_{\lambda}$. Thus we define AIC and BIC as follows:
\begin{eqnarray*}
   AIC(\lambda) &=& -2 Q(\widehat{\Theta}_{\lambda} \mid \widehat{\Theta}_{\lambda})  +  2 H(\widehat{\Theta}_{\lambda} \mid \widehat{\Theta}_{\lambda}) + 2 d , \nonumber \\
   BIC(\lambda)   &=&  -2 Q(\widehat{\Theta}_{\lambda} \mid \widehat{\Theta}_{\lambda})  +  2 H(\widehat{\Theta}_{\lambda} \mid \widehat{\Theta}_{\lambda}) + \log(n) d.
   \end{eqnarray*}
Then we choose the optimal value of the penalty parameter as the one that minimizes these criteria on a grid of candidate values for $\lambda$.

\noindent
\section{Analysis of data}
\label{sec:analysisofdata}
\setcounter{equation}{0}
 \setcounter{figure}{0}
 \setcounter{table}{0}

\subsection{Simulations}
\setcounter{equation}{0}
 \setcounter{figure}{0}
 \setcounter{table}{0}

We carried out simulation studies with a variety of data structures to compare how well competing methods recover the true graph structure. Though our EM approach is computationally expensive, we noticed that in our simulations the EM algorithm converges very quickly with a maximum of ten iterations and 100 MCMC samples for hundreds of variables.   

For the purpose of comparison we considered the following approaches:
\begin{itemize}
 \item[1.]  Proposed copula EM glasso (CopulaEM).
 \item[2.]  Proposed copula skeptic glasso (CopulaTau)
 \item[3.]  Nonparanormal normal-score based estimation with truncation presented in \citet{liu2012high} (NPNscore)
 \item[4.]  Nonparanormal skeptic using Kendall's tau presented in \citet{liu2012high} (NPNtau)
\end{itemize}

In our simulation we consider sample sizes (n=200) and number of variables(p=100) which are of mixed types that include binary(10), ordinal(10), count (10), nonnormal (eg. Chisquare (10)), and the remaining 60 are normal variables with outliers (none , 1\% , 20\%). In case of outliers, observations are replaced by a value either 5 or -5 with probability 0.6, see also \cite{liu2012high}. ROC curves are used to compare performance of the different approaches in recovering the true graph.   

Figures \ref{fig:roc1}  and \ref{fig:roc2} displays ROC curves based on averages of true positive rates and false positive rates computed from 100 times repeated simulations at each of 10 grid points of the tuning parameter. For mixed data with no and low level outliers, we see that the difference in the performance of recovering the true graph based on the various methods is negligible though our copula EM glasso shows slightly better performance in case of no outliers. 

\begin{figure}[tbh]
\centering
\begin{subfigure}
\centering
\includegraphics[width=5cm]{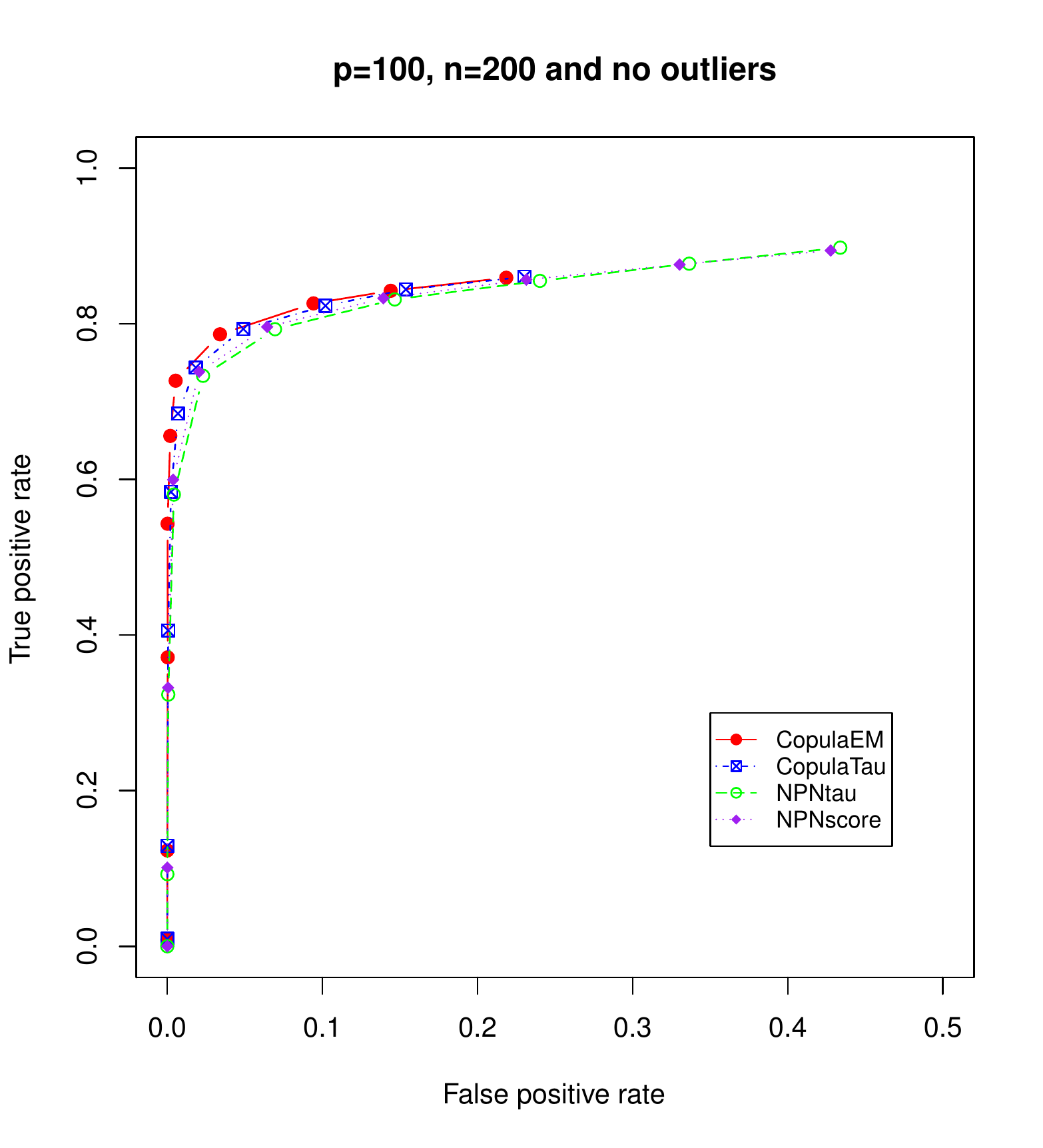} \\
\includegraphics[width=5cm]{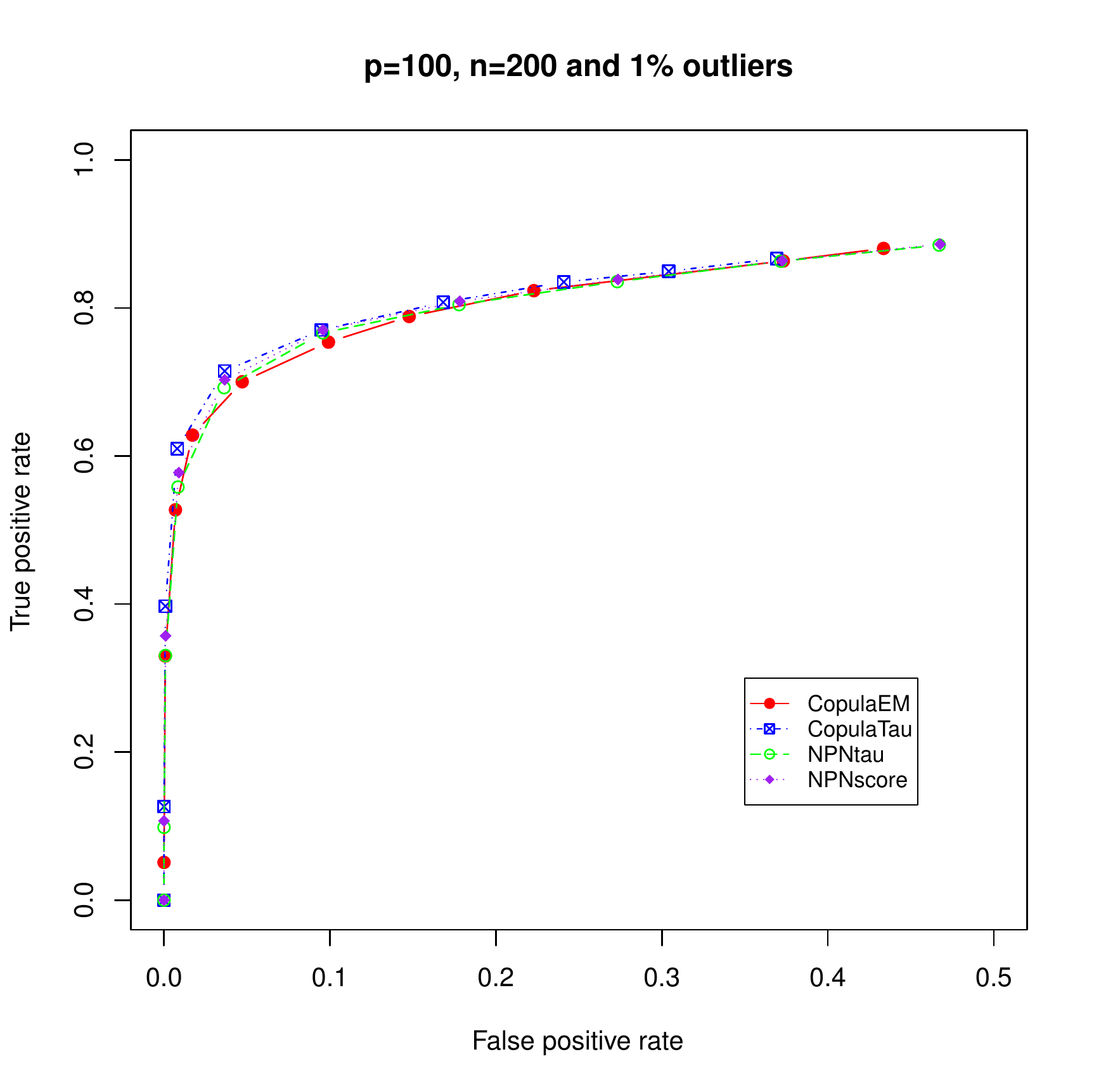} 
\end{subfigure}%
%\subfloat{\includegraphics[width=5cm]{mixedp100_n200_20pcexample.eps}}
\caption{\footnotesize{\it{ROC curves comparing various methods in recovering true graph structure for $n = 200$ and $p = 100$   in case of no and low level of outiers: our proposed approaches ``CopulaEM'' (copula EM glasso) and ``CopulaTau'' (copula skeptic) perform comparablly to that of ``NPNscore'' ( normal-score nonparanormal estimator) and ``NPNtau'' ( nonparanormal skeptic )}}}
\label{fig:roc1}
\end{figure}

In case of high level of outliers with mixed variables, the performance of the proposed copula skeptic and nonparanormal skeptic are comparable but the proposed copula skeptic out performs the nonparanormal skeptic. This suggests that a careful choice of parametric bivariate copulas results in better performance over the nonparametric approaches.

\begin{figure}[tbh]
	\centering
		\includegraphics[scale=0.5]{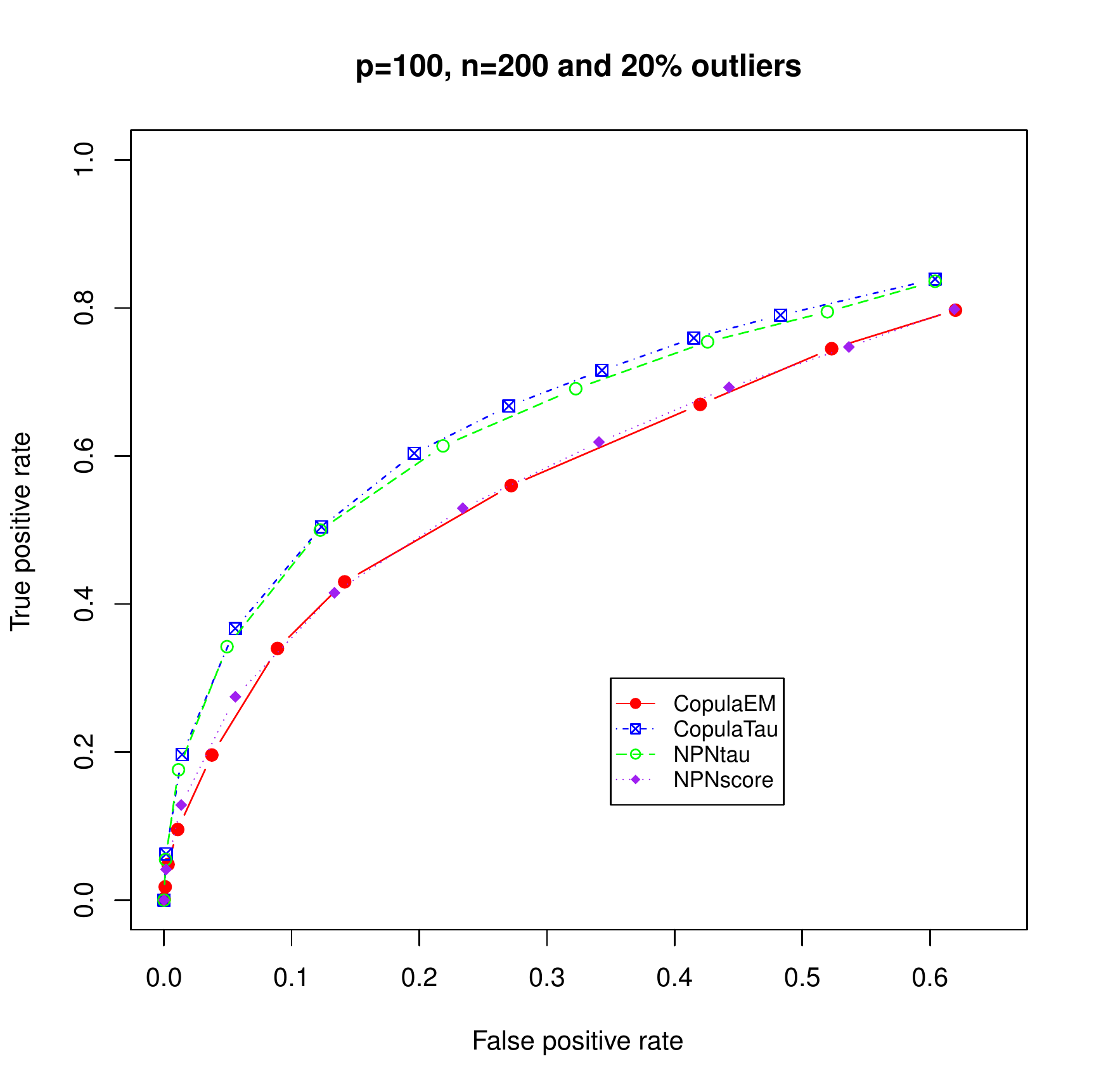} 
\caption{\footnotesize{\it{ROC curves comparing various methods in recovering true graph structure for $n = 200$ and $p = 100$  in the presence of high level of outliers: ``CopulaTau'' performs better than ``NPNtau'' while both out perferm  ``CopulaEM'' and ``NPNscore''}}}
	\label{fig:roc2}
\end{figure}

\subsection{Applications}
\setcounter{equation}{0}
\subsubsection{Breast cancer data}
In this section, we return to the motivation of our methodological development and apply the proposed Copula EM glasso  approach to the breast cancer data, which we introduced in Section \ref{sec:introduction}. The breast cancer experiment is a clinical study of DNA amplification and deletion patterns, using microarray technology. Its aim is to study the relationship between DNA amplification and deletion patterns (rather than gene expression) and the severity of the breast cancer, as measured by several clinical indicators on the patients. The data from the breast cancer experiment include 287 selected genes and 9 clinical variables obtained from 106 breast cancer patients. A brief description of clinical and genomic variables included in this study are presented in Table  \ref{tab:gencli}. 

\begin{table}[tbh]
\caption{\footnotesize List of genomic and clinical variables obtained from the breast cancer experiment. The aim is to find the underlying conditional dependence structure between these binary, count, ordered categorical, and continuous clinical and genomic variables.}
\begin{center} {\footnotesize}
 \begin{tabular}{l l l}
  \hline
  age.at.diagnosis 	&  age at diagnostics (in years) 	& continuous \\
  size   &  size of breast tumour (in mm) 	& continuous \\
  survival status		&  died due to breast cancer 	& binary \\
  grade     &  grade of breast cancer: 1 (low) to 3 (high)	& ordered categorical \\
  nodal stage   & lymp nodes involved & count \\
  NPI    &  NPI score 	& continuous \\
  ER & ER status: positive or negative & binary\\
  hist & Histology: Ductal or others & binary\\
  Ther & Therapy: Hormone or others & binary \\
  genes &   gene amplification/deletion &  continuous \\[1ex]
 \hline
\end{tabular}
\end{center}
\label{tab:gencli}
\end{table}
In the breast cancer study,  missing data rates among each of the gene amplification variables was less than 3\%. This could be that in microarray experiments it happens frequently that some part of the array could be damaged and resulted in some data to be excluded from consideration. Similarly, the missing data rates for the clinical variables were between 5\% and 20\%, respectively. 

In this study, we express the relationship between breast cancer survival, genomic and clinical variables as a series of conditional dependencies.  We applied the proposed Copula EM glasso described in this paper that internally samples missing observations. The BIC criterion resulted in an optimal penalty value of $\lambda=0.15$. A subnetwork of the complex dependence pattern among the observed variables induced by the underlying multivariate normal latent variables is displayed in Figure \ref{fig:breastred}. This subnetwork includes only links among the clinical and genomic variables. 

As can be seen from Figure \ref{fig:breastred}, breast cancer death is related to clinical variables (NPI score, Grade and size of breast cancer tumors) and  markers like SHGC4-207 and 10QTEL24.  As expected the NPI is directly related to breast cancer tumor size, cancer grade and nodes involved. The higher the values on the clinical variables the more aggressive the breast cancer and higher chance of death due to breast cancer.  

Further we see that the NPI score and the three clinical variables are related to the amplification or deletion of genes, for example,  BRCA1, RPS6KB1, ABL1, BMI1, CREBBP, STK6 (STK15), VHL, CTSB, PDGRL, GARP, ATM and PIM1. These findings are consistent with the literature that revealed these genes are associated with the risk and progression of breast cancer, see for example,  \citet{barlund2000detecting}, \citet{welcsh2001brca1}, \citet{dai2004synergistic}, \citet{srinivasan2006activation},  \citet{zia2007expression}, \citet{saeki2009diagnostic}, and \citet{rafn2012erbb2} among many others.

\begin{figure}[tbh]
	\centering
		\includegraphics[scale=0.5]{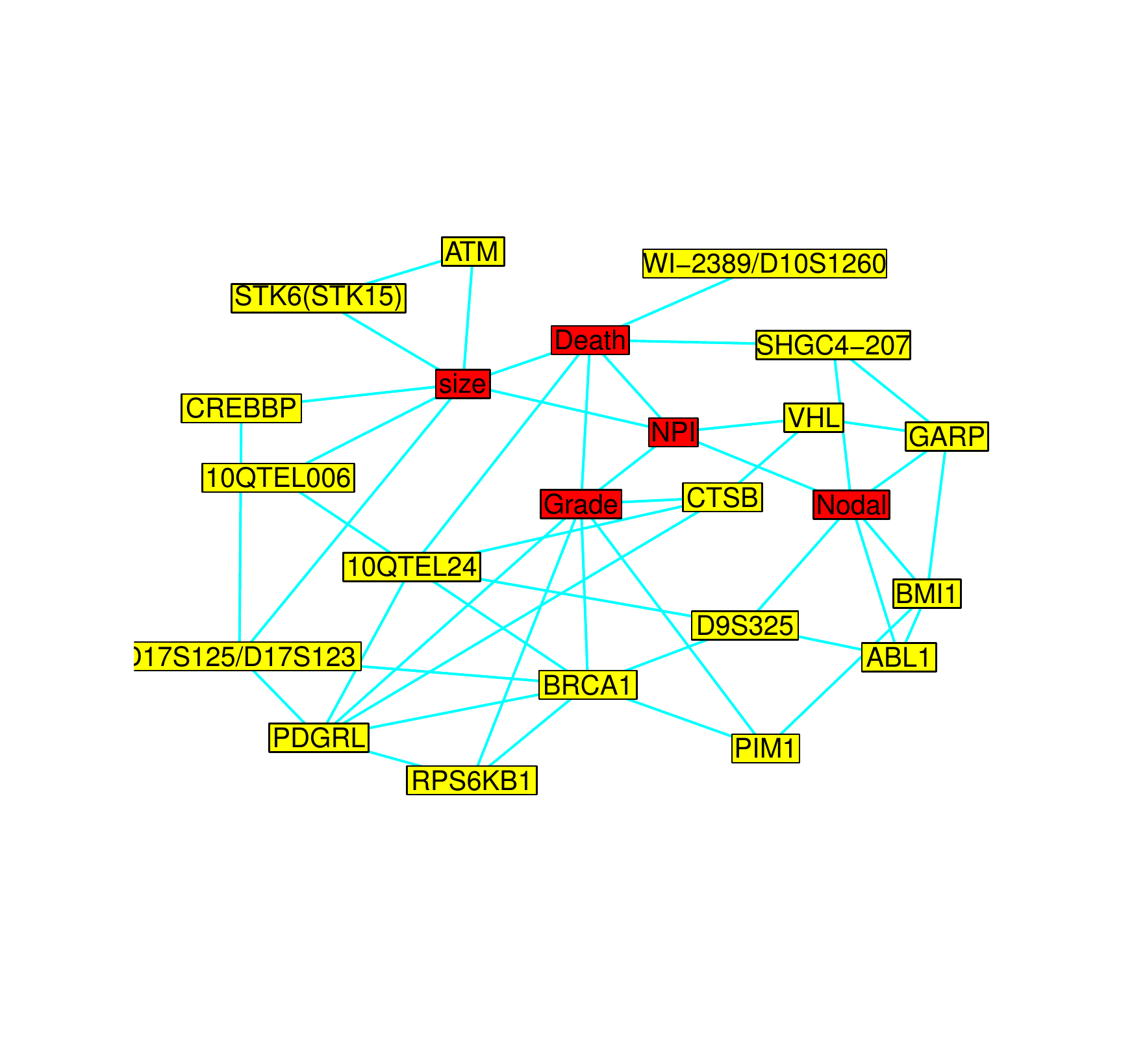} 
		\caption{\footnotesize{\it{ Conditional dependence subgraph of clinical variables and selected genes from the breast cancer data. Red color  or shaded rectangles represent clinical variables and yellow color or light shaded rectangles represent genomic variables.}}}
	\label{fig:breastred}
\end{figure}

\subsubsection{Maize genetic data}

 In this section, we consider data on maize genetic properties. The data from maize genetic nested association mapping population discussed by \citet{mcmullen2009genetic} is publicly available and downloaded from {http://www.panzea.org/lit/data\_sets.html}. The data contains 4699 samples or recombinant inbred lines combined across 25 families, representing 1106 SNP loci or genetic markers. For simplicity, to infer the genetic markers graph we treat the 4699 samples as replicates. The phenotypic variation measurements considered all reported by ordinal scale. Our objective is to identify trans-acting interactions of genetic markers across chromosomes in maize genome. The maize genome has 10 chromosomes. Trans-acting interactions also refered to as long-range chromosomal interactions or inter-chromosomal interactions has been studied, for example, in \citet{miele2008long} \citet{lum2011nonclassical}.

We applied copula skeptic glasso to the maize genetics data. Using minimum BIC criterion we obtained the value of the tuning parameter close to 0.05 taken in the range 0.03 (dense) to 0.20 (chromosome-wise separated) The resulting network is displayed in Figure \ref{fig:maize}. As expected we see from Figure \ref{fig:maize} that genetic markers within a chromosome are highly associated. On the other hand we see that  some genetic markers form links across chromosomes. These potential links, for example, are between PZA01601.1 (chromosome 8) and PZA02480.1 (chromosome 5); PZA00473.5 (chromosome 6) and PZA03624.1(chromosome 7); PZA02191.1 (chromosome 1) and PZA03321.4 (chromosome 2). These could refer trans-acting genetic markers which generate several interesting hypothesis for further experimental verifications. In support of our finding, \citet{mcmullen2009genetic} has also reported that among millions of pair-wise tests based on linkage disequilibrium (LD) marginally significant LD was observed between chromosome 6 and 7, though they concluded that it is a trivially small effect. We note that this could be possible in particular when a very large number of genetic markers are compared pair-wise, detecting even a single signiﬁcant pair-wise association is often hard because of the large multiple testing adjustment factor involved (see also \citet{buhlmann2013high}). The graphical modeling approach presented in this paper can be an efficient tool towards the study of interactions of genetic markers within (cis-acting) and across (trans-acting) chromosomes.

 \begin{figure}[tbh]
	\centering
		\includegraphics[scale=0.5]{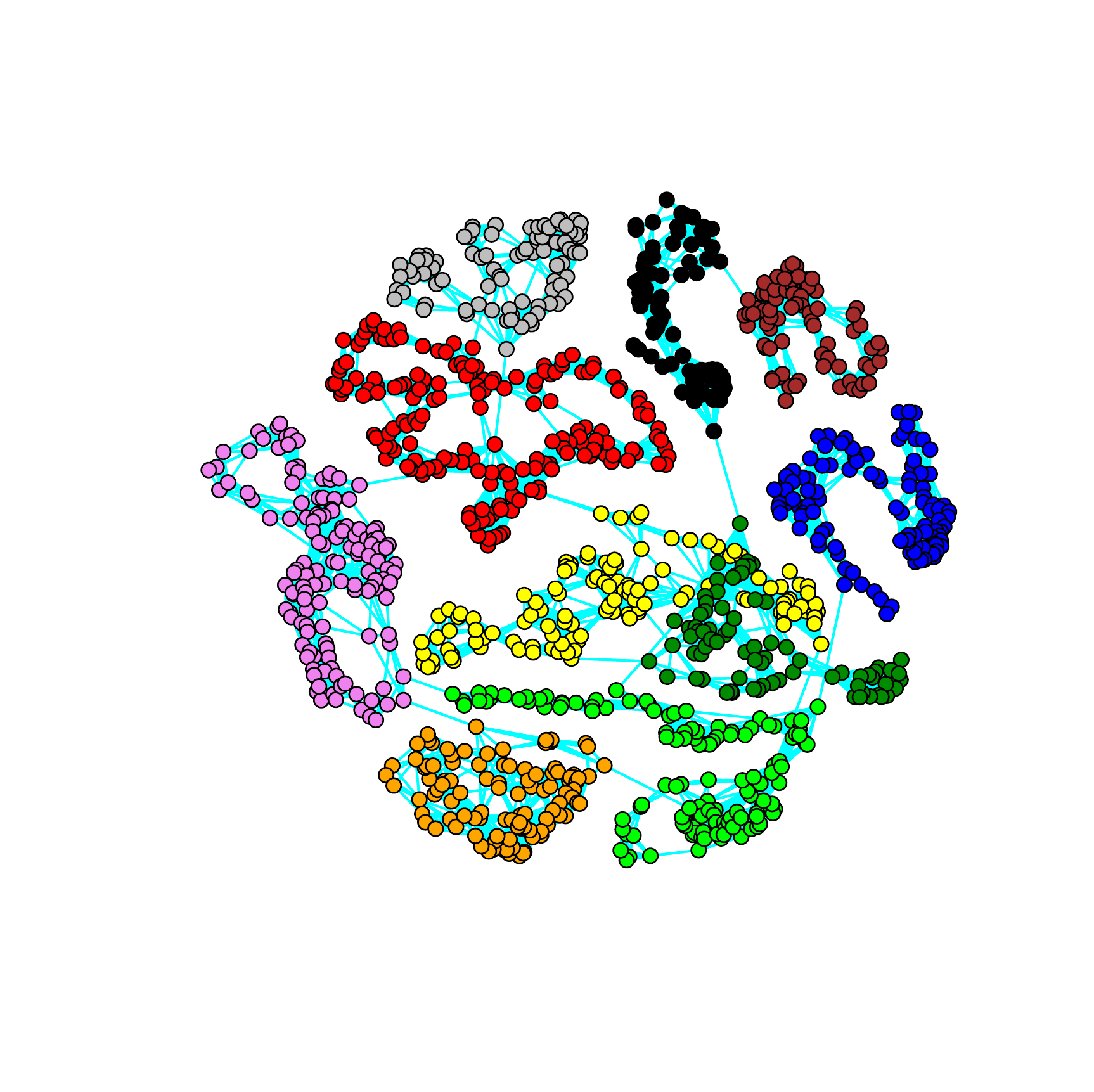} 
		\caption{\footnotesize {\it{Conditional dependence of genetic markers for the maize nested association mapping population. Genetic markers in each chromosome are represented by colours: red(Chr1), yellow(Chr2), green(Chr3), blue(Chr4), violet(Chr5), dark-green(Chr6), black(Chr7), orange(Chr8), grey(Chr9), and brown(Chr10)}}}
			\label{fig:maize}
\end{figure}

\section{Concluding Remarks}
\label{sec:conclusion}
Large high-dimensional datasets have become a common feature of many modern measurement techniques. In this article, we have presented a sparse copula Gaussian graphical model to infer networks from large high-dimensional data sets of arbitrary type. We proposed two approaches for the analysis of high dimensional mixed variables: $l_1$ penalized extended rank likelihood Gaussian copula based EM algorithm and copula skeptic glasso with pairwise parametric copula selection, not necessarily from the same family.

The performance of the proposed approaches in comparison to existing methods are evaluated using simulation studies. The simulation results suggest that the proposed copula EM glasso and copula skeptic glasso perform well to identifying  the true graph structure for high dimensional mixed variables setting. Taking into account computational efficiency we suggest to use the copula skeptic glasso for inferring networks  for very high dimensional (thousands of variables) and copula EM glasso for moderately high dimensional mixed variables setting. Moreover, the EM copula glasso approach has the advantage that it can be directly implemented for missing data without any additional computational issue.

We have illustrated the application of the proposed graphical modeling approaches on gene amplifications and deletions microarray data from breast cancer experiment and genetic markers from the maize nested association mapping population. We obtained a sparse representation of the conditional dependencies between the clinical and genetic variables, which generated several interesting hypotheses on the importance of these variables for the treatment of breast cancer. In particular, we identified many genes that are amplified or deleted in breast cancer and may functionally contribute to aggressiveness of breast cancer which is associated with worst outcome for the survival of breast cancer patients. The identification of such types of genes might lead to more accurate diagnostics and treatment at individual patient level. Similarly, a sparse representation of the interaction between genetic markers in maize genome, in particular across chromosomes will potentially be helpful for better understanding the molecular basis of phenotypic variation and to improve agricultural efficiency and sustainability. In general, the simulation and data analysis results suggest that the proposed copula based graphical modelings are promising approaches to infer networks for high dimensional nonnormal and mixed variables.

%\section*{Supplementary Materials}
%The computer code that can be used for the methodology developed in this paper is available online at  {\tt http://www.math.rug.nl/stat/Main/Research/} as an R statistical computing environment package, {\tt{CopulaGGM\_1.0.zip}}.
%
%\vspace*{-8pt}

\renewcommand{\bibname}{Bibliography}
\bibliographystyle{plainnat}
\bibliography{fentawgcem}

\end{document}